\documentclass[conference]{IEEEtran}
\usepackage{ifpdf}
\usepackage{cite}
\usepackage{tikz}
\usepackage{graphicx}
\usepackage{lipsum}
\usepackage{subfigure}
\usepackage{mathtools}
\usepackage{multirow}
\usepackage{cuted}
\usepackage{balance}
\usepackage[left=0.62in,right=0.62in,top=0.64in,bottom=0.7in]{geometry}
\usetikzlibrary{decorations, decorations.text,}
\usepackage{float}
\usepackage{amsmath}
\usepackage{amsthm}
\usepackage{mathtools, cuted}
\usepackage{array}
\usepackage{textcomp}
\usepackage{url}
\usepackage{epstopdf}
\usepackage{paralist,esint,amssymb,booktabs,cases,bm,galois}

\newcommand\Tx[1]{\mathrm{#1}}

\newcommand\Nm[1]{\lvert #1\rvert}

\newcommand{\RN}[1]{\textup{\uppercase\expandafter{\romannumeral#1}}}
\usepackage{cleveref}
\usetikzlibrary{calc}
\usetikzlibrary{shadings}
\newtheorem{theo}{Theorem}

\newtheorem{lemma}{Lemma}
\newtheorem{exam}{Example}
\newtheorem{defi}{Definition}

\makeatletter
\def\old@comma{,}
\catcode`\,=13
\def,{%
  \ifmmode%
    \old@comma\discretionary{}{}{}%
  \else%
    \old@comma%
  \fi%
}
\makeatother

\usepackage{algorithm}
\usepackage[noend]{algpseudocode}
\makeatletter
\def\BState{\State\hskip-\ALG@thistlm}
\makeatother

\definecolor{apple green}{rgb}{0.17,0.75,0.13}

\definecolor{edit}{rgb}{0.0,0.0,0.0}
\definecolor{edit_ah}{rgb}{0.0,0.0,0.0}

\errorcontextlines\maxdimen
\makeatletter
\newcommand*{\algrule}[1][\algorithmicindent]{\makebox[#1][l]{\hspace*{.5em}\vrule height 0.9 \baselineskip depth 0.3\baselineskip}}%

\newcount\ALG@printindent@tempcnta
\def\ALG@printindent{%
    \ifnum \theALG@nested>0
        \ifx\ALG@text\ALG@x@notext
            \addvspace{0pt}
        \else
            \unskip
            \ALG@printindent@tempcnta=1
            \loop
                \algrule[\csname ALG@ind@\the\ALG@printindent@tempcnta\endcsname]%
                \advance \ALG@printindent@tempcnta 1
            \ifnum \ALG@printindent@tempcnta<\numexpr\theALG@nested+1\relax
            \repeat
        \fi
    \fi
    }%
\usepackage{etoolbox}
\patchcmd{\ALG@doentity}{\noindent\hskip\ALG@tlm}{\ALG@printindent}{}{\errmessage{failed to patch}}
\makeatother



\begin{document}
\bstctlcite{IEEEexample:BSTcontrol}
\title{\vspace{-0.45em}GRADE-AO: Towards Near-Optimal Spatially-Coupled Codes With High Memories\vspace{-0.25em}}
\author{\IEEEauthorblockN{Siyi Yang\IEEEauthorrefmark{1}, Ahmed Hareedy\IEEEauthorrefmark{2}, Shyam Venkatasubramanian\IEEEauthorrefmark{1}, Robert Calderbank\IEEEauthorrefmark{2}, and Lara Dolecek\IEEEauthorrefmark{1}}
\IEEEauthorblockA{\IEEEauthorrefmark{1}Electrical and Computer Engineering Department, University of California, Los Angeles, Los Angeles, CA 90095 USA\\
\IEEEauthorrefmark{2}Electrical and Computer Engineering Department, Duke University, Durham, NC 27708 USA\\
siyiyang@ucla.edu, ahmed.hareedy@duke.edu, shyam1999@ucla.edu, robert.calderbank@duke.edu, and dolecek@ee.ucla.edu\vspace{-0.4em}
}}
\maketitle

\begin{abstract}
Spatially-coupled (SC) codes, known for their threshold saturation phenomenon and low-latency windowed decoding algorithms, are ideal for streaming applications. They also find application in various data storage systems because of their excellent performance. SC codes are constructed by partitioning an underlying block code, followed by rearranging and concatenating the partitioned components in a ``convolutional'' manner. The number of partitioned components determines the ``memory'' of SC codes. While adopting higher memories results in improved SC code performance, obtaining optimal SC codes with high memory is known to be hard. In this paper, we investigate the relation between the performance of SC codes and the density distribution of partitioning matrices. We propose a probabilistic framework that obtains (locally) optimal density distributions via gradient descent. Starting from random partitioning matrices abiding by the obtained distribution, we perform low complexity optimization algorithms over the cycle properties to construct high memory, high performance quasi-cyclic SC codes. Simulation results show that codes obtained through our proposed method notably outperform state-of-the-art SC codes with the same constraint length and codes with uniform partitioning. 

\end{abstract}

\IEEEpeerreviewmaketitle

\section{Introduction}
\label{sectoin: introduction}

Spatially-coupled (SC) codes, also known as low-density parity-check (LDPC) codes with convolutional structures, are an ideal choice for streaming applications and storage devices thanks to their threshold saturation phenomenon \cite{5695130,kumar2014threshold,olmos2015scaling,hareedy2017high,lentmaier2010iterative} and amenability to low-latency windowed decoding \cite{Iyengar2013windowed}. SC codes are constructed by partitioning the parity-check matrix of an underlying block code, followed by rearranging the component matrices in a ``convolutional'' manner. In particular, component matrices are concatenated into a ``replica'', and then multiple replicas are placed together, resulting in a ``coupled'' code. The number of component matrices minus one is referred to as the ``memory'' of the SC codes \cite{mitchell2015spatially,esfahanizadeh2018finite,hareedy2020channel,pusane2011deriving}.

It is known that the performance of an SC code improves as its memory increases. This is a byproduct of improved node expansion and additional degrees of freedom that can be utilized to decrease the number of small cycles and detrimental objects \cite{esfahanizadeh2018finite,hareedy2020channel,dolecek2010analysis,naseri2020spatially}. Although the optimization problem of designing SC codes with memory less than $4$ has been efficiently solved \cite{esfahanizadeh2018finite,hareedy2020channel}, there remains a vacuum in efficient algorithms that construct good enough SC codes with high memories. Esfahanizadeh et al. \cite{esfahanizadeh2018finite} proposed a combinatorial framework to develop optimal quasi-cyclic (QC) SC codes, comprising so-called optimal-overlap (OO) to search for the optimal partitioning matrices, and circulant power optimization (CPO) to optimize the lifting parameters, which was extended by Hareedy et al. \cite{hareedy2020channel}. However, this method is hard to execute in practice for high memory codes due to the increasing computational complexity. Battaglioni et al. developed an algorithmic method that searches for good SC codes with high memories \cite{battaglioni2017design}. However, high memory codes designed by purely algorithmic methods are unable to offer strict guarantees on performance superiority; several of these codes can even be beat by optimally designed QC-SC codes with lower memories under the same constraint length. Therefore, a method that theoretically identifies an avenue to a near-optimal construction of SC codes with high memories is of significant interest. 

In a way similar to random coding in spirit, our objective is to obtain some near-optimal solutions starting from~a random partitioning matrix, where the density distribution of component matrices (i.e., edge distribution) is analogous to the degree distribution in random coding. While discrete optimization methods \cite{esfahanizadeh2018finite,hareedy2020channel} have been shown to suffer from exponential growths in complexity, fueled by the increase in degrees of freedom, we adopt a more efficient, probabilistic framework that searches for the optimal edge distribution via gradient descent, referred to as \textbf{gradient-descent distributor (GRADE)}, followed by an \textbf{algorithmic optimizer (AO)} that obtains a locally optimal partition near a random partition with edge distribution obtained from GRADE. The current goal is still to minimize the number of small cycles, which reduces undesirable dependencies, and thus improves the code performance. The impact of this probabilistic method extends beyond its performance gains and low complexity. Particularly in the error floor region, a more advanced set of detrimental objects (absorbing sets \cite{dolecek2010analysis}) governs the LDPC code performance. Our probabilistic method also has high potential to be extended to handle detrimental objects specified by the channel.

In this paper, we propose a probabilistic framework that efficiently searches for near-optimal SC codes with high memories. In \Cref{section: preliminaries}, we introduce preliminaries of SC codes and the performance-related metrics. In \Cref{section: framework}, we develop the theoretical basis of GRADE, which derives a locally optimal edge distribution from an arbitrarily provided initial distribution and conditions. In \Cref{section: construction}, we introduce two examples of GRADE-AOs that result in near-optimal SC codes: the so-called \textbf{gradient descent (GD) codes} and \textbf{topologically-coupled (TC) codes}. Our proposed framework is validated in \Cref{section: simulation} by simulation results of four groups of codes, with the best code in each obtained from GRADE-AO. Finally, we make concluding remarks and introduce possible future work in \Cref{section: conclusion}.

\section{Preliminaries}
\label{section: preliminaries}

In this section, we recall the typical construction of SC codes with QC structure. Any QC code with a parity-check matrix $\mathbf{H}$ is obtained by replacing each nonzero (zero) entry of some binary matrix $\mathbf{H}^{\textup{P}}$ with a circulant (zero) matrix of size $z$, $z\in\mathbb{N}$. The matrix $\mathbf{H}^{\textup{P}}$ and $z$ are referred to as the protograph and the circulant size of the code, respectively. In particular, the protograph $\mathbf{H}^{\textup{P}}_{\textup{SC}}$ of an SC code has a convolutional structure composed of $L$ replicas, as presented in Fig.~\ref{fig: SC protograph}. Each replica is obtained by stacking the disjoint component matrices $\{\mathbf{H}^{\textup{P}}_i\}_{i=0}^{m}$, where $m$ is the memory and $\bm{\Pi}=\mathbf{H}^{\textup{P}}_{0}+\mathbf{H}^{\textup{P}}_1+\cdots+\mathbf{H}^{\textup{P}}_m$ is the protograph of the underlying block code. 

In this paper, we constrain $\bm{\Pi}$ to be an all-one matrix of size $\gamma\times \kappa$, $\gamma,\kappa\in\mathbb{N}$. An SC code is then uniquely represented by its partitioning matrix $\mathbf{P}$ and lifting matrix $\mathbf{L}$, where $\mathbf{P}$ and $\mathbf{L}$ are all $\gamma\times \kappa$ matrices. The matrix $\mathbf{P}$ has $(\mathbf{P})_{i,j}=a$ if $(\mathbf{H}^{\textup{P}}_a)_{i,j}=1$. The matrix $\mathbf{L}$ is determined by replacing each circulant matrix by its associated exponent. Here, this exponent represents the power to which the matrix $\bm{\sigma}$ defined by $(\bm{\sigma})_{i,i+1}=1$ is raised, where $(\bm{\sigma})_{z,z+1}=(\bm{\sigma})_{z,1}$.

\begin{figure}
\centering
\includegraphics[width=0.4\textwidth]{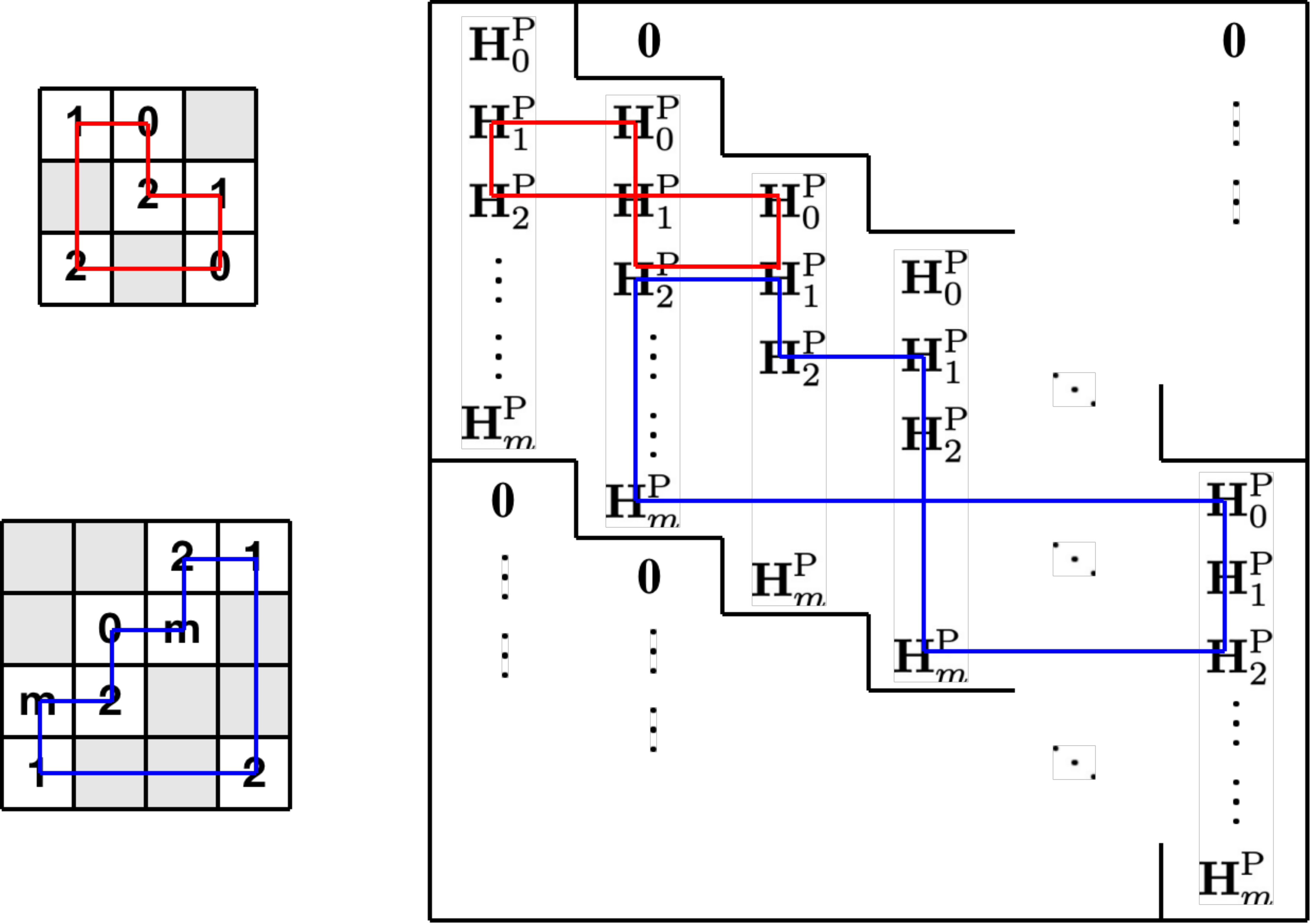}
\caption{Cycles in the protograph (right panel) and their corresponding structures in the partitioning matrices (left panel).}
\label{fig: SC protograph}
\end{figure}

The performance of finite-length LDPC codes is strongly affected by the number of detrimental objects that are subgraphs with certain structures in the Tanner graphs of those codes. Two major classes of detrimental objects are trapping sets and absorbing sets. Since enumerating and minimizing the number of detrimental objects is complicated, existing work typically focuses on common substructures of these objects: the small cycles \cite{esfahanizadeh2018finite,hareedy2020channel,battaglioni2017design}. A cycle-$2g$ candidate in $\mathbf{H}_{\textup{SC}}^{\textup{P}}$ ($\bm{\Pi}$) is a path of traversing a structure to generate cycles of length $2g$ after lifting (partitioning) \cite{hareedy2020channel}. In an SC code, each cycle in the Tanner graph corresponds to a cycle candidate in the protograph $\mathbf{H}_{\textup{SC}}^{\textup{P}}$, and each cycle candidate in $\mathbf{H}_{\textup{SC}}^{\textup{P}}$ corresponds to a cycle candidate $C$ in the base matrix $\mathbf{\Pi}$. \Cref{lemma: cycle condition} specifies a necessary and sufficient condition for a cycle candidate in $\bm{\Pi}$ to become a cycle candidate in the protograph and then a cycle in the final Tanner graph.

\begin{lemma}\label{lemma: cycle condition} Let $C$ be a cycle-$2g$ candidate in the base matrix, where $g\in\mathbb{N}$, $g\geq 2$. Denote $C$ by $(j_1,i_1,j_2,i_2,\dots,j_g,i_g)$, where $(i_{k},j_{k})$, $(i_{k},j_{k+1})$, $1\leq k\leq g$, $j_{g+1}=j_1$, are nodes of $C$ in $\bm{\Pi}$, $\mathbf{P}$, and $\mathbf{L}$. Then $C$ becomes a cycle candidate in the protograph if and only if the following condition follows \cite{fossorier2004quasicyclic}:
\begin{equation}\label{eqn: partition cycle}
\sum\nolimits_{k=1}^{g}\mathbf{P}(i_{k},j_{k})=\sum\nolimits_{k=1}^{g}\mathbf{P}(i_{k},j_{k+1}).
\end{equation}
This cycle candidate becomes a cycle in the Tanner graph if and only if:
\begin{equation}\label{eqn: final cycle}
\sum\nolimits_{k=1}^{g}\mathbf{L}(i_{k},j_{k})\equiv\sum\nolimits_{k=1}^{g}\mathbf{L}(i_{k},j_{k+1}) \mod z.
\end{equation}
\end{lemma}

As shown in Fig.~\ref{fig: SC protograph}, a cycle-$6$ candidate and a cycle-$8$ candidate in the partitioning matrix with assignments satisfying condition (\ref{eqn: partition cycle}), and their corresponding cycle candidates in the protograph are marked by red and blue, respectively. An optimization of a QC-SC code is typically divided into two major steps: optimizing $\mathbf{P}$ to minimize the number of cycle candidates in the protograph, and optimizing $\mathbf{L}$ to further reduce that number in the Tanner graph given the optimized $\mathbf{P}$ \cite{esfahanizadeh2018finite,hareedy2020channel}. The latter goal has been achieved in \cite{esfahanizadeh2018finite} and \cite{hareedy2020channel}, using an algorithmic method called circulant power optimization (CPO), while the former goal is yet to be achieved for large $m$. We note that the step separation highlighted above notably reduces the overall optimization complexity.

In the remainder of this paper, we focus on SC codes for the additive white Gaussian noise (AWGN) channel, where the most detrimental objects are the low weight absorbing sets \cite{esfahanizadeh2018finite}. Consequently, a simplified optimization focuses on cycle candidates of lengths $4$, $6$, and $8$ \cite{esfahanizadeh2018finite,hareedy2020channel}. Existing literature shows that the optimal $\mathbf{P}$ for an SC code with $m\leq 2$ typically has a balanced (uniform) edge distribution among component matrices \cite{esfahanizadeh2018finite}. However, in the remaining sections, we show that the edge distribution for optimal SC codes with large $m$ is not uniform, and we propose the GRADE-AO framework that explores a locally optimal solution.

\section{A Probabilistic Optimization Framework}
\label{section: framework}

In this section, we present a probabilistic framework that searches for a locally optimal edge distribution for the partitioning matrices of SC codes with given memories through the gradient descent algorithm. 

\begin{defi}\label{defi: coupling pattern} Let $\gamma,\kappa,m,m_t\in\mathbb{N}$ and $\mathbf{a}=\left(a_0,a_1,\dots,a_{m_t}\right)$, where $0=a_0<a_1<\cdots<a_{m_t}=m$. A $(\gamma,\kappa)$ SC code with memory $m$ is said to have \textbf{coupling pattern} $\mathbf{a}$ if and only if $\mathbf{H}^{\textup{P}}_i\neq\mathbf{0}^{\gamma\times \kappa}$, for all $i\in\{a_0,a_1,\dots,a_{m_t}\}$, and $\mathbf{H}^{\textup{P}}_i=\mathbf{0}^{\gamma\times \kappa}$, otherwise. The value $m_t$ is called the \textbf{pseudo-memory} of the SC code.
\end{defi}

\subsection{Probabilistic Metric}
\label{subsec: metric}
In this subsection, we define metrics linking the edge distribution with the expected number of cycle candidates in the protograph in \Cref{theo: cycle 6 probability} and \Cref{theo: cycle 8 probability}. While Schmalen et al. have shown in \cite{Schmalen} that nonuniform coupling (nonuniform edge distribution in our paper) yields an improved threshold, our work differs in two areas: 1) Explicit optimal coupling graphs were exhaustively searched and were restricted to small memories in \cite{Schmalen}, whereas our method produces near-optimal SC protographs for arbitrary memories. 2) Work \cite{Schmalen} focused on the asymptotic analysis for the threshold region, while our framework is dedicated to the finite-length construction and has additional demonstrable gains in the error floor region.

\begin{defi}\label{defi: coupling polynommial}
Let $m, m_t\in\mathbb{N}$ and $\mathbf{a}=\left(a_0,a_1,\dots,a_{m_t}\right)$, where $0=a_0<a_1<\cdots<a_{m_t}=m$. Let $\mathbf{p}=\left(p_0,p_1\dots,p_{m_t}\right)$, where $0<p_i\leq 1$, $p_0+p_1+\cdots+p_{m_t}=1$. Then, the following $f(X;\mathbf{a},\mathbf{p})$, which is abbreviated to $f(X)$ when the context is clear, is called the \textbf{coupling polynomial} of an SC code with coupling pattern $\mathbf{a}$, associated with probability distribution~$\mathbf{p}$:
\begin{equation}\label{eqn: coupling polynomial}
f(X;\mathbf{a},\mathbf{p})\triangleq \sum\nolimits_{0\leq i\leq m_t} p_i X^{a_i}.
\end{equation}
\end{defi}

\begin{theo}\label{theo: cycle 6 probability} Let $\left[\cdot\right]_i$ denote the coefficient of $X^{i}$ of a polynomial. Denote by $P_{6}(\mathbf{a},\mathbf{p})$ the probability of a cycle-$6$ candidate in the base matrix becoming a cycle-$6$ candidate in the protograph under random partitioning with edge distribution $\mathbf{p}$. Then,
\begin{equation} 
P_{6}(\mathbf{a},\mathbf{p})=\left[f^3(X)f^3(X^{-1})\right]_0.
\end{equation}
\end{theo}

\begin{proof} According to \Cref{lemma: cycle condition}, suppose the cycle-$6$ candidate in the base matrix is represented by $C(j_1,i_1,j_2,i_2,j_3,i_3)$. Then,
\begin{equation*}\label{eqn: equation}
\begin{split}
&P_{6}(\mathbf{a},\mathbf{p})=\mathbb{P}\left[\sum\nolimits_{k=1}^{3}\mathbf{P}(i_{k},j_{k})=\sum\nolimits_{k=1}^{3}\mathbf{P}(i_{k},j_{k+1})\right]\\
=&\sum\limits_{\sum\nolimits_{k=1}^{3}x_i=\sum\nolimits_{k=1}^{3}y_i} \mathbb{P}\left[ \mathbf{P}(i_{k},j_{k})=x_k,\mathbf{P}(i_{k},j_{k+1})=y_k \right]\\
=&\sum\limits_{\sum\nolimits_{k=1}^{3}x_i=\sum\nolimits_{k=1}^{3}y_i} p_{x_1}p_{x_2}p_{x_3}p_{y_1}p_{y_2}p_{y_3}\\
=&\left[\sum\limits_{x_i,y_i\in \Tx{supp}(\mathbf{a})} p_{x_1}p_{x_2}p_{x_3}p_{y_1}p_{y_2}p_{y_3} X^{x_1+x_2+x_3-y_1-y_2-y_3}\right]_0\\
=&\left[f^3(X)f^3(X^{-1})\right]_0.\\
\end{split}
\end{equation*}
The theorem is proved.
\end{proof}

\begin{exam}\label{exam: SC probability} Consider SC codes with full memories and uniform partition, i.e., $\mathbf{a}=(0,1,\dots,m)$ and $\mathbf{p}=\frac{1}{m+1} \mathbf{1}_{m+1}$. When $m=2$, $P_6(\mathbf{a},\mathbf{p})=0.1934$; when $m=4$, $P_6(\mathbf{a},\mathbf{p})=0.1121$.
\end{exam}

\begin{exam}\label{exam: SC_m2_opt_edge_dist} First, consider SC codes with $m=m_t=2$. Let $\mathbf{a}_1=(0,1,2)$ and $\mathbf{p}_1=(2/5,1/5,2/5)$. According to \Cref{theo: cycle 6 probability}, $f(X)=(2+X+2X^2)/5$, $f^3(X)f^3(X^{-1})=0.0041(X^6+X^{-6})+0.0123(X^5+X^{-5})+0.0399(X^4+X^{-4})+0.0717(X^3+X^{-3})+0.1267(X^2+X^{-2})+0.1544(X+X^{-1})+0.1818$. Therefore, $P_6(\mathbf{a}_1,\mathbf{p}_1)=0.1818$. Second, consider SC codes with $m=m_t=4$. Let $\mathbf{a}_2=(0,1,2,3,4)$ and $\mathbf{p}_2=(0.31,0.13,0.12,0.13,0.31)$. According to \Cref{theo: cycle 6 probability}, $P_6(\mathbf{a}_2,\mathbf{p}_2)=0.0986$.
\end{exam}

\begin{figure}
\centering
\includegraphics[width=0.49\textwidth]{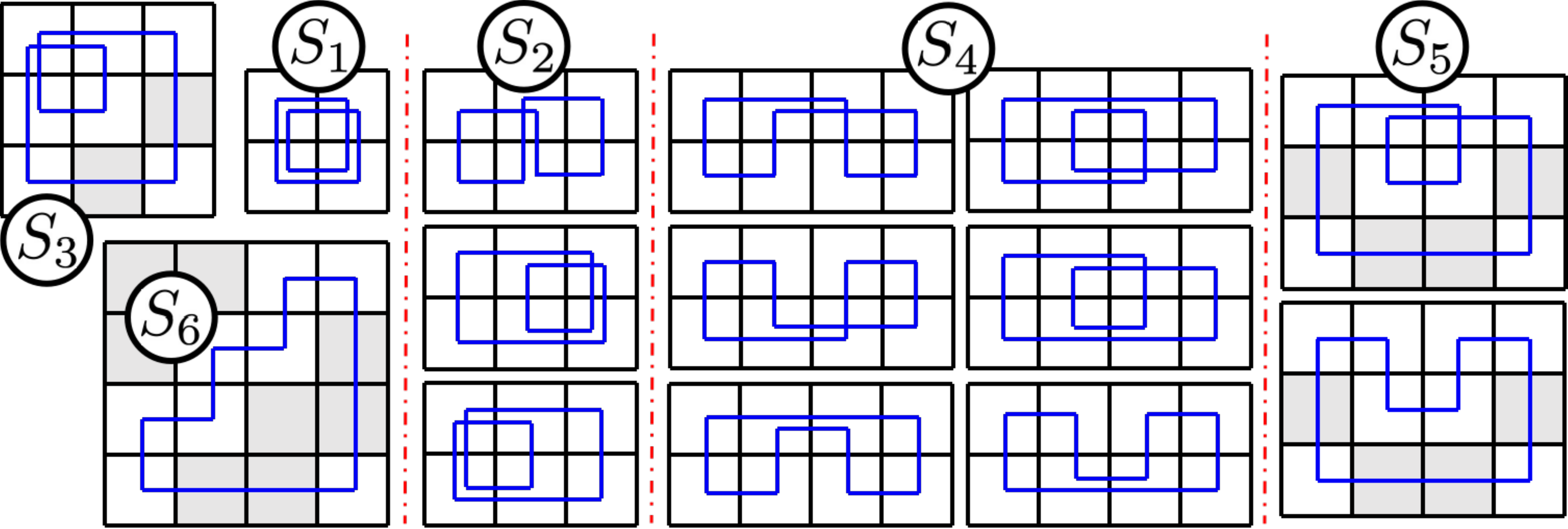}\vspace{-0.3em}
\caption{Structures and cycle candidates for cycle 8.}
\label{fig: cycle8}
\vspace{-0.2em}
\end{figure}

After we have derived the metric for cycle-$6$ candidates in the protograph, we now turn to the case of cycle-$8$ candidates. As shown in Fig.~\ref{fig: cycle8}, cycle candidates in the base matrix that result in cycle-$8$ candidates in the protograph can be categorized into $6$ different structures, labeled $S_1,\dots,S_6$. Different cases are differentiated by the number of rows and columns (without order) the structures span in the partitioning matrix \cite{hareedy2020channel}. Specifically, $S_1,\dots,S_6$ denote the structures that span submatrices of size $2\times 2$, $2\times 3$ or $3\times 2$, $3\times 3$, $2\times 4$ or $4\times 2$, $3\times 4$ or $4\times 3$, and $4\times 4$, respectively. Any structure that belongs to $S_2,S_4, S_5$ has multiple cycle-$8$ candidates, and these distinct candidates are marked by blue in Fig.~\ref{fig: cycle8}. 

\begin{lemma}\label{lemma: cycle 8 pattern probability} Denote $P_{8;i}(\mathbf{a},\mathbf{p})$, $1\leq i\leq 6$, as the average probability of a cycle-$8$ candidate of structure $S_i$ in the base matrix becoming a cycle-$8$ candidate in the protograph, under random partition with edge distribution $\mathbf{p}$. Then,
\begin{equation}\label{eqn: lemma cycle 8 pattern probability}
\begin{split}
\hspace{-0.5em}P_{8;1}(\mathbf{a},\mathbf{p})&=\left[f^2(X)f^2(X^{-1})\right]_0,\\
\hspace{-0.5em}P_{8;2}(\mathbf{a},\mathbf{p})&=\left[f(X^2)f(X^{-2})f^2(X)f^2(X^{-1})\right]_0,\\
\hspace{-0.5em}P_{8;3}(\mathbf{a},\mathbf{p})&=\left[f(X^2)f^2(X)f^4(X^{-1})\right]_0,\text{ and}\\
\hspace{-0.5em}P_{8;4}(\mathbf{a},\mathbf{p})&=P_{8;5}(\mathbf{a},\mathbf{p})=P_{8;6}(\mathbf{a},\mathbf{p})=\left[f^4(X)f^4(X^{-1})\right]_0. \nonumber
\end{split}
\end{equation}
\end{lemma}

\begin{proof} For patterns where the nodes of the cycle-$8$ candidates are pairwise different, namely, $S_4,S_5,S_6$, the result can be derived by following the logic in the proof of \Cref{theo: cycle 6 probability}. 

For $S_1$, suppose the indices of the rows and columns are $i_1,i_2$, and $j_1,j_2$, respectively. Then, the cycle condition in \Cref{lemma: cycle condition} is $\mathbf{P}(i_1,j_1)+\mathbf{P}(i_2,j_2)=\mathbf{P}(i_1,j_2)+\mathbf{P}(i_2,j_1)$. 

For $S_2$, suppose the indices of the rows and columns are $i_1,i_2$, and $j_1,j_2,j_3$, respectively. Then, the cycle condition in \Cref{lemma: cycle condition} is $2\mathbf{P}(i_1,j_1)-2\mathbf{P}(i_2,j_1)+\mathbf{P}(i_2,j_2)+\mathbf{P}(i_2,j_3)-\mathbf{P}(i_1,j_2)-\mathbf{P}(i_1,j_3)=0$. 

For $S_3$, suppose the indices of the rows and columns are $i_1,i_2,i_3$, and $j_1,j_2,j_3$, respectively. Then, the cycle condition in \Cref{lemma: cycle condition} is $2\mathbf{P}(i_1,j_1)+\mathbf{P}(i_2,j_2)+\mathbf{P}(i_3,j_3)-\mathbf{P}(i_1,j_2)-\mathbf{P}(i_2,j_1)-\mathbf{P}(i_1,j_3)-\mathbf{P}(i_3,j_1)=0$. 

Following the logic in the proof of \Cref{theo: cycle 6 probability}, the case for $S_1,S_2,S_3$ can be proved.
\end{proof}

\begin{theo}\label{theo: cycle 8 probability} Denote $N_8(\mathbf{a},\mathbf{p})$ as the expectation of the number of cycle-$8$ candidates in the protograph. Then,
\begin{equation}\label{eqn: lemma cycle 8 probability}
\begin{split}
N_{8}(\mathbf{a},\mathbf{p})&=w_1\left[f^2(X)f^2(X^{-1})\right]_0\\
&+w_2\left[f(X^2)f(X^{-2})f^2(X)f^2(X^{-1})\right]_0\\
&+w_3\left[f(X^2)f^2(X)f^4(X^{-1})\right]_0\\
&+w_4\left[f^4(X)f^4(X^{-1})\right]_0,\\
\end{split}
\end{equation}
where $w_1=\binom{\gamma}{2}\binom{\kappa}{2}$, $w_2=3\binom{\gamma}{2}\binom{\kappa}{3}+3\binom{\gamma}{3}\binom{\kappa}{2}$, $w_3=18\binom{\gamma}{3}\binom{\kappa}{3}$, $w_4=6\binom{\gamma}{2}\binom{\kappa}{4}+6\binom{\gamma}{4}\binom{\kappa}{2}+36\binom{\gamma}{3}\binom{\kappa}{4}+36\binom{\gamma}{4}\binom{\kappa}{3}+24\binom{\gamma}{4}\binom{\kappa}{4}$.
\end{theo}

\begin{proof} Provided the results in \Cref{lemma: cycle 8 pattern probability}, we just need to prove that the numbers of cycle candidates of structures $S_1,S_2,\dots,S_6$ in a $\gamma\times \kappa$ base matrix are $\binom{\gamma}{2}\binom{\kappa}{2}$, $3\binom{\gamma}{2}\binom{\kappa}{3}+3\binom{\gamma}{3}\binom{\kappa}{2}$, $18\binom{\gamma}{3}\binom{\kappa}{3}$, $6\binom{\gamma}{2}\binom{\kappa}{4}+6\binom{\gamma}{4}\binom{\kappa}{2}$, $36\binom{\gamma}{3}\binom{\kappa}{4}+36\binom{\gamma}{4}\binom{\kappa}{3}$, and $24\binom{\gamma}{4}\binom{\kappa}{4}$, respectively. 

Take $i=5$ as an example, the number of cycle candidates of structure $S_5$ in any $3\times 4$ or $4\times 3$ matrix is $3\cdot \binom{4}{2} \cdot 2=36$. The total number of $3\times 4$ or $4\times 3$ matrices in a $\gamma\times \kappa$ base matrix is $\binom{\gamma}{3}\binom{\kappa}{4}+\binom{\gamma}{4}\binom{\kappa}{3}$. Therefore, the total number of cycles of pattern $S_5$ is $36\binom{\gamma}{3}\binom{\kappa}{4}+36\binom{\gamma}{4}\binom{\kappa}{3}$. By a similar logic, we can prove the result for the remaining patterns.
\end{proof}

\subsection{Gradient Descent Distributor}
\label{subsec: gradient descent}
By contrasting Examples \ref{exam: SC probability} with \ref{exam: SC_m2_opt_edge_dist} it is clear that for a given coupling pattern, an optimal edge distribution is not necessarily reached by a uniform partition. In this subsection, we develop an algorithm that obtains a locally optimal distribution by gradient descent.

\begin{lemma}\label{lemma: cycle 6 distribution} Given $m_t\in\mathbb{N}$ and $\mathbf{a}=(a_0,a_1,\dots,a_{m_t})$, a necessary condition for $P_6(\mathbf{a},\mathbf{p})$ to reach its minimum value is that the following equation holds for some $c_0\in\mathbb{R}$:
\begin{equation}\label{eqn: cycle 6 condition}
\left[f^3(X)f^2(X^{-1})\right]_{a_i}=c_0,\ \forall i, 0\leq i\leq m_t.
\end{equation}
\end{lemma}

\begin{proof} Consider the gradient of $L_6(\mathbf{a},\mathbf{p})=P_6(\mathbf{a},\mathbf{p})+c(1-p_{0}-p_{1}-\dots-p_{{m_t}})$. 
\begin{equation}\label{eqn: lemma cycle 6 gradient}
\begin{split}
&\nabla_{\mathbf{p}}L_6(\mathbf{a},\mathbf{p})\\
=&\nabla_{\mathbf{p}}\left(P_6(\mathbf{a},\mathbf{p})+c(1-p_{0}-p_{1}-\dots-p_{{m_t}})\right)\\
=&\nabla_{\mathbf{p}}\left[f^3(X)f^3(X^{-1})\right]_0-c\mathbf{1}_{m_t+1}\\
=&\left[\nabla_{\mathbf{p}}\left(f^3(X)f^3(X^{-1})\right)\right]_0-c\mathbf{1}_{m_t+1}\\
=&3\left[f^2(X)f^2(X^{-1})f(X)\nabla_{\mathbf{p}}f(X^{-1})\right]_0\\
&+3\left[f^2(X)f^2(X^{-1})f(X^{-1})\nabla_{\mathbf{p}}f(X)\right]_0-c\mathbf{1}_{m_t+1}\\
=&6\left[f^3(X)f^2(X^{-1})\left(X^{-a_{0}},X^{-a_{1}},\dots,X^{-a_{m_t}}\right)\right]_0-c\mathbf{1}_{m_t+1}.
\end{split}
\end{equation}
When $P_8(\mathbf{a},\mathbf{p})$ reaches its minimum, $\nabla_{\mathbf{p}}\left[L(\mathbf{a},\mathbf{p})\right]=\mathbf{0}_{m_t+1}$, which is equivalent to (\ref{eqn: cycle 6 condition}) by defining $c_0=c/6$.
\end{proof}

\begin{lemma}\label{lemma: cycle 8 distribution} Given $\gamma,\kappa,m_t\in\mathbb{N}$ and $\mathbf{a}=(a_0,a_1,\dots,a_{m_t})$, a necessary condition for $N_8(\mathbf{a},\mathbf{p})$ to reach its minimum value is that the following equation holds for some $c_0\in\mathbb{R}$:
\begin{equation*}\label{eqn: cycle 8 condition}
\begin{split}
&\left[4f^2(X)f(X^{-1})\right])_{a_i}+\bar{w}_2\left[2f(X^{2})f^2(X)f^2(X^{-1})\right]_{2a_i}\\
+&\bar{w}_2\left[4f(X^2)f(X^{-2})f^2(X)f(X^{-1})\right]_{a_i}\\
+&\bar{w}_3\left[f^2(X)f^4(X^{-1})\right]_{-2a_i}+\bar{w}_3\left[2f(X^2)f(X)f^4(X^{-1})\right]_{-a_i}\\
+&\bar{w}_3\left[4f(X^2)f^2(X)f^3(X^{-1})\right]_{a_i}\\
+&\bar{w}_4\left[8f^4(X)f^3(X^{-1})\right]_{a_i}=c_0,\ \forall i, 0\leq i\leq m_t,
\end{split}
\end{equation*}
where $\bar{w}_2=\gamma+\kappa-4$, $\bar{w}_3=2(\gamma-2)(\kappa-2)$, and $\bar{w}_4=\frac{1}{2}\left[(\gamma-2)(\gamma-3)+(\kappa-2)(\kappa-3)\right]+(\gamma-2)(\kappa-2)(\gamma+\kappa-6)+\frac{1}{6}(\gamma-2)(\gamma-3)(\kappa-2)(\kappa-3)$.
\end{lemma}

\begin{proof} Consider the gradient of $L_8(\mathbf{a},\mathbf{p})=N_8(\mathbf{a},\mathbf{p})+c(1-p_{0}-p_{1}-\dots-p_{{m_t}})$. 
\begin{equation}\label{eqn: lemma cycle 8 gradient}
\begin{split}
&\nabla_{\mathbf{p}}L_8 (\mathbf{a},\mathbf{p})\\
=&\nabla_{\mathbf{p}}\left(N_8(\mathbf{a},\mathbf{p})+c(1-p_{0}-p_{1}-\dots-p_{{m_t}})\right)\\
=&w_1\left[\nabla_{\mathbf{p}}\left(f^2(X)f^2(X^{-1})\right)\right]_0\\
&+w_2\left[\nabla_{\mathbf{p}}\left(f(X^2)f(X^{-2})f^2(X)f^2(X^{-1})\right)\right]_0\\
&+w_3\left[\nabla_{\mathbf{p}}\left(f(X^2)f^2(X)f^4(X^{-1})\right)\right]_0\\
&+w_4\left[\nabla_{\mathbf{p}}\left(f^4(X)f^4(X^{-1})\right)\right]_0-c\mathbf{1}_{m_t+1}\\
=&w_1\{\left[4f^2(X)f(X^{-1})(X^{-a_0},X^{-a_1},\dots,X^{-a_{m_t}})\right])_0\\
&+\bar{w}_2\left[2f(X^2)f^2(X)f^2(X^{-1})(X^{-2a_0},\dots,X^{-2a_{m_t}})\right]_0\\
&+\bar{w}_2\left[4f(X^2)f(X^{-2})f^2(X)f(X^{-1})(X^{-a_0},\dots,X^{-a_{m_t}})\right]_0\\
&+\bar{w}_3\left[f^2(X)f^4(X^{-1})(X^{2a_0},X^{2a_1},\dots,X^{2a_{m_t}})\right]_0\\
&+\bar{w}_3\left[2f(X^2)f(X)f^4(X^{-1})(X^{a_0},X^{a_1},\dots,X^{a_{m_t}})\right]_0\\
&+\bar{w}_3\left[4f(X^2)f^2(X)f^3(X^{-1})(X^{-a_0},X^{-a_1},\dots,X^{-a_{m_t}})\right]_0\\
&+\bar{w}_4\left[8f^4(X)f^3(X^{-1})(X^{-a_0},X^{-a_1},\dots,X^{-a_{m_t}})\right]_0\}\\
&-c\mathbf{1}_{m_t+1}.
\end{split}
\end{equation}
When $P_8(\mathbf{a},\mathbf{p})$ reaches its minimum, $\nabla_{\mathbf{p}}\left[L(\mathbf{a},\mathbf{p})\right]=\mathbf{0}_{m_t+1}$, which is equivalent to (\ref{eqn: cycle 8 condition}) by defining $c_0=c/w_1$.
\end{proof}

Based on \Cref{lemma: cycle 6 distribution} and \Cref{lemma: cycle 8 distribution}, we adopt the gradient descent algorithm to obtain a locally optimal edge distribution for SC codes with coupling pattern $\mathbf{a}$, starting from the uniform distribution inside $\mathbf{P}$ as presented in \Cref{algo: GRADE}. Note that $\Tx{conv}(\cdot)$ and $\Tx{inverse}(\cdot)$ refer to convolution and reverse of vectors, respectively.

\begin{algorithm}
\caption{Gradient Descent Distributor (GRADE)} \label{algo: GRADE}
\begin{algorithmic}[1]
\Require 
\Statex $\gamma,\kappa,m_t,m,\mathbf{a}$: parameters of the SC code;
\Statex $w$: weight of each cycle-$6$ candidate;
\Statex $\epsilon,\alpha$: accuracy and step size of gradient descent;
\Ensure
\Statex $\mathbf{p}$: a locally optimal edge distribution over $\Tx{supp}(\mathbf{a})$;
\State $\bar{w}_1\gets \frac{2w}{3}(\gamma-2)(\kappa-2)$, obtain $\{\bar{w}_i\}_{i=2}^4$ in \Cref{lemma: cycle 8 distribution}; 
\State $v_{prev}=1$; $v_{cur}=1$;
\State $\mathbf{p},\mathbf{g}\gets \mathbf{0}_{m_t+1}$, $\mathbf{f},\bar{\mathbf{f}}\gets \mathbf{0}_{m+1}$, $\mathbf{f}_2,\bar{\mathbf{f}}_2\gets \mathbf{0}_{2m+1}$;
\State $\mathbf{p}\gets \frac{1}{m_t+1} \mathbf{1}_{m_t+1}$;
\State $\mathbf{f}\left[a_0,\dots,a_{m_t}\right]\gets\mathbf{p}$, $\bar{\mathbf{f}}\gets \Tx{inverse}(\mathbf{f})$;
\State $\mathbf{f}_2\left[1,3,\dots,2m+1\right]\gets\mathbf{f}$, $\bar{\mathbf{f}}_2\gets \Tx{inverse}(\mathbf{f}_2)$;
\State $\mathbf{q}_1 \gets \bar{w}_1\Tx{conv}(\mathbf{f},\mathbf{f},\mathbf{f},\bar{\mathbf{f}},\bar{\mathbf{f}},\bar{\mathbf{f}})$, $\mathbf{q}_2\gets \Tx{conv}(\mathbf{f},\mathbf{f},\bar{\mathbf{f}},\bar{\mathbf{f}})$;
\State $\mathbf{q}_3\gets \bar{w}_2\Tx{conv}(\mathbf{f}_2,\bar{\mathbf{f}}_2,\mathbf{f},\mathbf{f},\bar{\mathbf{f}},\bar{\mathbf{f}})+\bar{w}_3\Tx{conv}(\mathbf{f}_2,\mathbf{f},\mathbf{f},\bar{\mathbf{f}},\bar{\mathbf{f}},\bar{\mathbf{f}},\bar{\mathbf{f}})$
\Statex \hspace{13pt}$+\bar{w}_4\Tx{conv}(\mathbf{f},\mathbf{f},\mathbf{f},\mathbf{f},\bar{\mathbf{f}},\bar{\mathbf{f}},\bar{\mathbf{f}},\bar{\mathbf{f}})$;
\State $v_{prev}=v_{cur}$, $v_{cur}=\mathbf{q}_1\left[3m\right]+\mathbf{q}_2\left[2m\right]+\mathbf{q}_3\left[4m\right]$;
\State $\mathbf{g}_1\gets 6\bar{w}_1\Tx{conv}(\mathbf{f},\mathbf{f},\mathbf{f},\bar{\mathbf{f}},\bar{\mathbf{f}})$, $\mathbf{g}_2\gets 4\Tx{conv}(\mathbf{f},\mathbf{f},\bar{\mathbf{f}})$;
\State $\mathbf{g}_3\gets 4\bar{w}_2\Tx{conv}(\mathbf{f}_2,\bar{\mathbf{f}}_2,\mathbf{f},\mathbf{f},\bar{\mathbf{f}})+2\bar{w}_3\Tx{conv}(\bar{\mathbf{f}}_2,\mathbf{f},\mathbf{f},\mathbf{f},\mathbf{f},\bar{\mathbf{f}})$
\Statex \hspace{13pt}$+4\bar{w}_3\Tx{conv}(\mathbf{f}_2,\mathbf{f},\mathbf{f},\bar{\mathbf{f}},\bar{\mathbf{f}},\bar{\mathbf{f}})+8\bar{w}_4\Tx{conv}(\mathbf{f},\mathbf{f},\mathbf{f},\mathbf{f},\bar{\mathbf{f}},\bar{\mathbf{f}},\bar{\mathbf{f}})$;
\State $\mathbf{g}_4\gets 2\bar{w}_2\Tx{conv}(\mathbf{f}_2,\mathbf{f},\mathbf{f},\bar{\mathbf{f}},\bar{\mathbf{f}})+\bar{w}_3\Tx{conv}(\mathbf{f},\mathbf{f},\mathbf{f},\mathbf{f},\bar{\mathbf{f}},\bar{\mathbf{f}})$;
\State $\mathbf{g}\gets \mathbf{g}_1\left[2m+\mathbf{a}\right]+\mathbf{g}_2\left[m+\mathbf{a}\right]+\mathbf{g}_3\left[3m+\mathbf{a}\right]$
\Statex \hspace{13pt}$+\mathbf{g}_4\left[2m+2\mathbf{a}\right]$, $\mathbf{g}\gets \mathbf{g}-\Tx{mean}(\mathbf{g})$;
\If{$|v_{prev}-v_{cur}|>\epsilon$} 
\State $\mathbf{p}\gets \mathbf{p}-\alpha\frac{\mathbf{g}}{||\mathbf{g}||}$;
\State \textbf{goto} step 5;
\EndIf
\State \textbf{return} $\mathbf{p}$;
\end{algorithmic}
\end{algorithm}

\section{Construction}
\label{section: construction}

In this section, we present two algorithmic optimization methods based on GRADE to obtain locally optimal SC codes with a fixed coupling pattern.

\subsection{Gradient Descent Codes}
\label{subsec: GD}

In this subsection, we consider SC codes with full memories, i.e., $m=m_t$ and $\mathbf{a}=(0,1,\dots,m)$. In this case, our proposed GRADE algorithm obtains a highly skewed edge distribution. Starting from a random partitioning matrix $\mathbf{P}$ with the derived distribution, one can perform a semi-greedy algorithm that searches for the locally optimal partitioning matrix near the initial $\mathbf{P}$. Constraining the search space to contain $\mathbf{P}$'s that have distributions within small $L_1$ and $L_\infty$ distances from that of the original $\mathbf{P}$, and adopting the CPO next, significantly reduces the computational complexity to find a strong high memory code. This procedure is given in \Cref{algo: AO}.

\begin{algorithm}
\caption{GRADE-A Optimizer (AO)} \label{algo: AO}
\begin{algorithmic}[1]
\Require 
\Statex $\gamma,\kappa,m$: parameters of an SC code with full memory;
\Statex $\mathbf{p}$: edge distribution obtained from \Cref{algo: GRADE};
\Statex $d_1$, $d_2$: parameters indicating the size of the searching space;
\Ensure
\Statex $\mathbf{P}$: a locally optimal partitioning matrix;
\State Obtain the lists $\mathcal{L}_6(i,j)$, $\mathcal{L}_8(i,j)$ of cycle-$6$ candidates and cycle-$8$ candidates in the base matrix that contain node $(i,j)$, $1\leq i\leq \gamma$, $1\leq j\leq\kappa$;
\State Obtain $\mathbf{u}=\Tx{arg}\min\nolimits_{\mathbf{x}\in \mathbb{N}^{m+1}, ||\mathbf{x}||_1=\gamma\kappa}||\frac{1}{\gamma\kappa}\mathbf{x}-\mathbf{p}||_2$;
\For{$i\in\{0,1,\dots,m\}$} 
\State Place $\mathbf{u}\left[i+1\right]$ $i$'s into $\mathbf{P}$ randomly;
\EndFor
\State $\mathbf{d}\gets \mathbf{0}_{m+1}$;
\For{$i\in\{1,2,\dots,\gamma\}$, $j\in\{1,2,\dots,\kappa\}$}
\State $noptimal\gets \Tx{False}$;
\State $n_6\gets \Nm{\mathcal{L}_6(i,j)}$, $n_8\gets \Nm{\mathcal{L}_8(i,j)}$, $n\gets wn_6+n_8$;
\For{$v\in\{0,1,\dots,z-1\}$}
\State $\mathbf{d}'=\mathbf{d}$, $\mathbf{d}'\left[v+1\right]\gets \mathbf{d}'\left[v+1\right]+1$, $p\gets \mathbf{P}(i,j)$;
\If{$||\mathbf{d}'||_1\leq d_1$ and $||\mathbf{d}'||_{\infty}\leq d_2$}
\State $\mathbf{P}(i,j)\gets v$;
\State $t_6\gets \Nm{\mathcal{L}_6(i,j)}$, $t_8\gets \Nm{\mathcal{L}_8(i,j)}$, $t\gets wt_6+t_8$;
\If{$t<n$}
\State $noptimal\gets\Tx{True}$, $n\gets t$, $\mathbf{d}\gets\mathbf{d}'$;
\Else 
\State $\mathbf{P}(i,j)\gets p$;
\EndIf
\EndIf
\EndFor
\EndFor
\If{$noptimal$} 
\State \textbf{goto} step 6;
\EndIf
\State \textbf{return} $\mathbf{P}$;
\end{algorithmic}
\end{algorithm}

We refer to codes obtained from \Cref{algo: AO} as \textbf{gradient descent (GD) codes}. By replacing the input distribution $\mathbf{p}$ in \Cref{algo: AO} with the uniform distribution, we obtain the so-called \textbf{uniform (UNF) codes}. We show in \Cref{section: simulation} by simulation that the distribution obtained by GRADE results in constructions that are better than those adopting uniform distribution and in existing literature.

\subsection{Topologically-Coupled Codes}
\label{subsec: TC}

In this subsection, we explore SC codes with pseudo-memory $m_t$ such that $m_t < m$ and $\mathbf{a}\neq (0,1,\dots,m)$. The motivation behind this task is to construct an SC code with memory $m$ with the same computational complexity needed to construct a memory $m_t$ code, where $m_t<m$. Given $m_t$ and $m$, we first find the optimal $\mathbf{a}$ with length $m_t+1$ in a brute-force manner. Taking $m=4$ and $m_t=2$ as an example, the~optimal coupling pattern is reached by $\mathbf{a}=(0,1,4)$ and the corresponding optimal distribution is almost uniform.

Given the optimal coupling pattern $\mathbf{a}$, we then obtain an optimal partitioning matrix by the OO method proposed in \cite{esfahanizadeh2018finite} and \cite{hareedy2020channel}. We extend the OO method for memory $m_0$ SC codes to any SC code with pseudo-memory $m_t=m_0$, which does not increase the complexity of the approach. We refer to the codes obtained from the extended OO method followed by the CPO as \textbf{topologically-coupled (TC) codes}. 

Optimal TC codes with pseudo-memory $m_t$ have strictly fewer cycle candidates in their protographs than optimal SC codes with full memory $m=m_t$. Take $m=4$ and $m_t=2$ as an example. Suppose the optimal SC code has the partition $\bm{\Pi}=\mathbf{H}_0^{\Tx{P}}+\mathbf{H}_1^{\Tx{P}}+\mathbf{H}_2^{\Tx{P}}$. Consider the TC code with partition $\bm{\Pi}=\mathbf{H}_0^{\Tx{P}}+\mathbf{H}_1^{\Tx{P}}+\mathbf{H}_4^{\Tx{P}}$ such that $\mathbf{H}_2^{\Tx{P}}=\mathbf{H}_4^{\Tx{P}}$. Then, any cycle-$6$ candidate resulting from a cycle candidate in the base matrix assigned with $0$-$1$-$0$-$1$-$2$-$0$, $1$-$2$-$1$-$2$-$2$-$0$, or $0$-$1$-$2$-$1$-$x$-$x$, $x\in\{0,1,2\}$, in $\mathbf{P}$ no longer has a counterpart in the TC code, since by replacing $2$'s with $4$'s, assignments $0$-$1$-$0$-$1$-$4$-$0$, $1$-$4$-$1$-$4$-$4$-$0$, and $0$-$1$-$4$-$1$-$x$-$x$, $x\in\{0,1,4\}$, no longer satisfy the cycle condition in \Cref{lemma: cycle condition}. Moreover, there exists a bijection between the remaining candidates in the SC code and all candidates in the TC code through the replacement of $2$'s with $4$'s. Therefore, TC codes are better (have less cycles) than SC codes with the same circulant size. In \Cref{section: simulation}, we present simulation results of such codes and show that they can also outperform SC codes with the same constraint length (larger circulant size).

\section{Simulation Results}
\label{section: simulation}

In this section, we obtain the frame error rate (FER) curves of four groups of SC codes designed by the GRADE-AO methods presented in \Cref{section: construction}. We show that codes constructed by the GRADE-AO methods offer significant performance gains compared with codes with uniform edge distributions and codes constructed through purely algorithmic methods.

Out of these three plots, the left and center ones compare GD codes with UNF codes designed as in \Cref{subsec: GD}. The right plot compares a TC code designed as in \Cref{subsec: TC} with optimal SC codes constructed through the OO-CPO method proposed in \cite{esfahanizadeh2018finite}. The GD/UNF codes have parameters $(\gamma,\kappa,m,z,L)=(3,7,5,13,100)$, $(3,17,9,7,100)$, and $(4,29,19,29,20)$, respectively. The TC code has parameters $(\gamma,\kappa,m_t,z,L)=(4,17,2,17,50)$ with coupling pattern $\mathbf{a}=(0,1,4)$. For a fair comparison, we have selected two SC codes: one with a similar constraint length $(m+1)z$ and the other with an identical circulant power $z$. To ensure that the SC codes and the TC code have close rates and codelengths, the two SC codes have parameters $(\gamma,\kappa,m,z,L)=(4,17,2,28,30)$ and $(4,17,2,17,50)$, respectively. The statistics regarding the number of cycles of each code are presented in \Cref{table: cycle statistics}.

\begin{table}
\label{fig: partition_gamma4}
\centering
\caption{Statistics of the Number of Cycles}
\begin{tabular}{c|c|c|c}
\toprule
 $(\gamma,\kappa)$ & Code & Cycles-$6$ & Cycles-$8$ \\
\hline
\multirow{2}{*}{$(3,7)$} & GD & 0 & 0\\
 & UNF & 0 & 6292\\
 \hline
\multirow{3}{*}{$(3,17)$} & GD & 0 & 397880\\
 & UNF & 0 & 559902 \\
 & Battaglioni et al.\cite{battaglioni2017design}& 0 & 451337\\
 \hline
\multirow{2}{*}{$(4,29)$} & GD & 0 & 528,090\\
 & UNF & 0 & 1,087,268\\
 \hline
\multirow{3}{*}{$(4,17)$} & TC & 15,436 & -\\
 & SC (matched constraint length) & 19,180 & -\\
 & SC (matched circulant size)& 74,579 & -\\
\bottomrule
\end{tabular}
\label{table: cycle statistics}
\end{table}

Fig.~\ref{fig: FER_GD_UNF_3} shows FER curves of our GD/UNF comparisons with $(\gamma,\kappa)=(3,7)$ and $(3,17)$. The partitioning matrices and the lifting matrices of the codes are specified in Fig.~\ref{fig: code_GD_UNF_3_7}-\ref{fig: code_GD_UNF_3_17}. When $\gamma=3$, cycles-$6$ are easily removed by the CPO. Therefore, we perform joint optimization on the number of cycle-$6$ and cycle-$8$ candidates by assigning different weights to cycle candidates in \Cref{algo: AO}. We observe a performance gain for the GD code with respect to the UNF code in both the waterfall region and the error floor region. Moreover, the number of cycles-$8$ in the $(3,17)$ GD code is reduced by $29\%$ and $12\%$ compared with the UNF code and the code constructed by Battaglioni et al. in \cite{battaglioni2017design}, respectively. In addition, the $(3,17)$ GD code has no weight-$6$ absorbing sets (ASs) and $133$ weight-$7$ ASs, whereas the UNF code has $6$ weight-$6$ ASs and $361$ weight-$7$ ASs. As for the $(3,7)$ codes, all cycles-$6$ and cycles-$8$ are removed. Thus, the gain of the GD code compared with the UNF code exceeds the gain observed in the $(3,17)$ codes.

\begin{figure}
\centering
  \subfigure[GD code.]{%
       \includegraphics[width=0.35\textwidth]{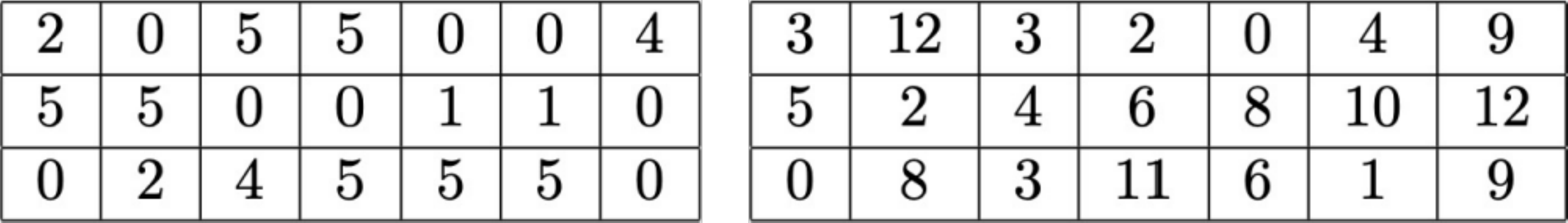}
       \label{1a}
       }
  \subfigure[UNF code.]{%
        \includegraphics[width=0.35\textwidth]{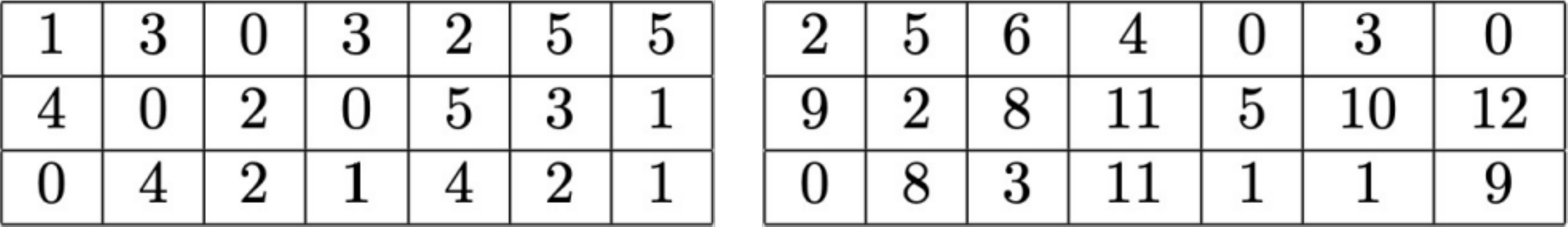}
        \label{1b}
        }
  \caption{Partitioning matrices (left) and circulant power matrices (right) of GD/UNF codes with $(\gamma,\kappa,m,z,L)=(3,7,5,13,100)$.}
  \label{fig: code_GD_UNF_3_7} 
\end{figure}

\begin{figure}
\centering
  \subfigure[GD code.]{%
       \includegraphics[width=0.38\textwidth]{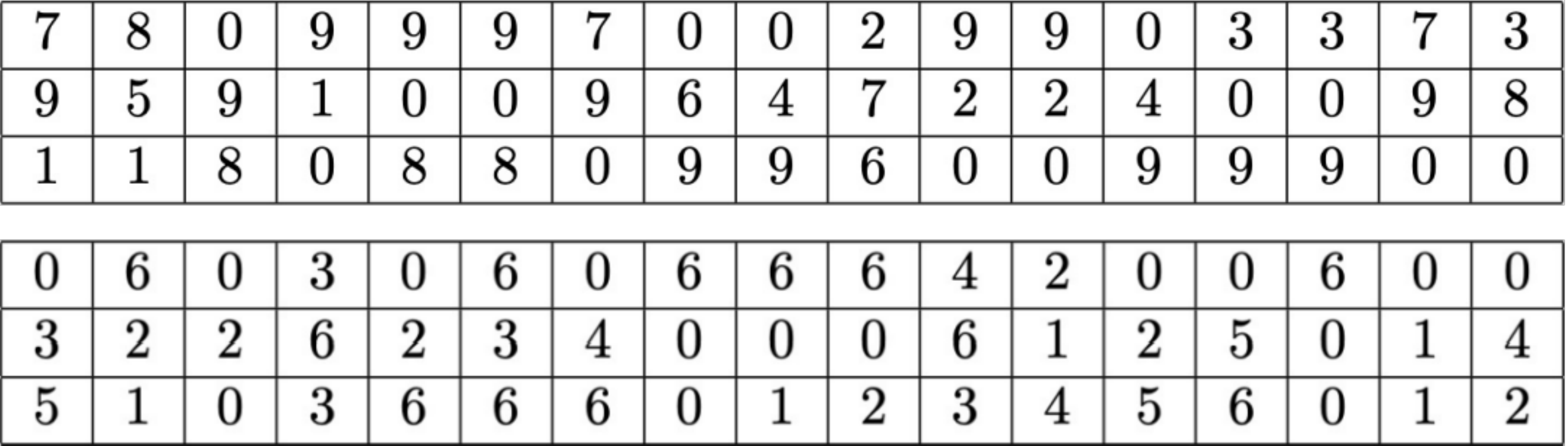}
       \label{1a}
       }
  \subfigure[UNF code.]{%
        \includegraphics[width=0.38\textwidth]{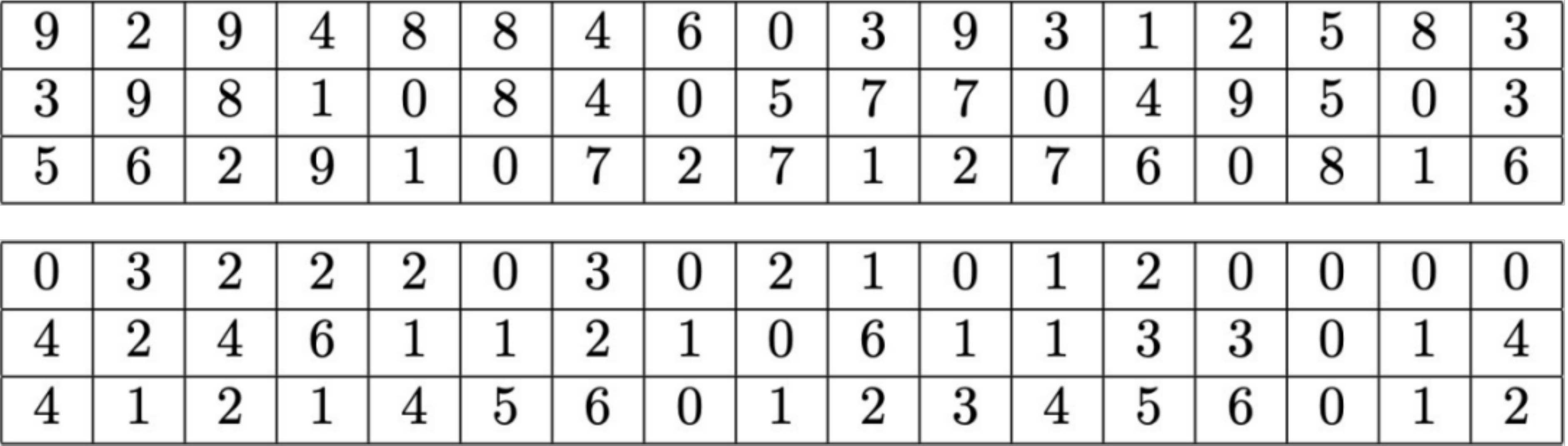}
        \label{1b}
        }
  \caption{Partitioning matrices (top) and circulant power matrices (bottom) of GD/UNF codes with $(\gamma,\kappa,m,z,L)=(3,17,9,7,100)$.}
  \label{fig: code_GD_UNF_3_17} 
\end{figure}

\begin{figure}
\centering
\includegraphics[width=0.5\textwidth]{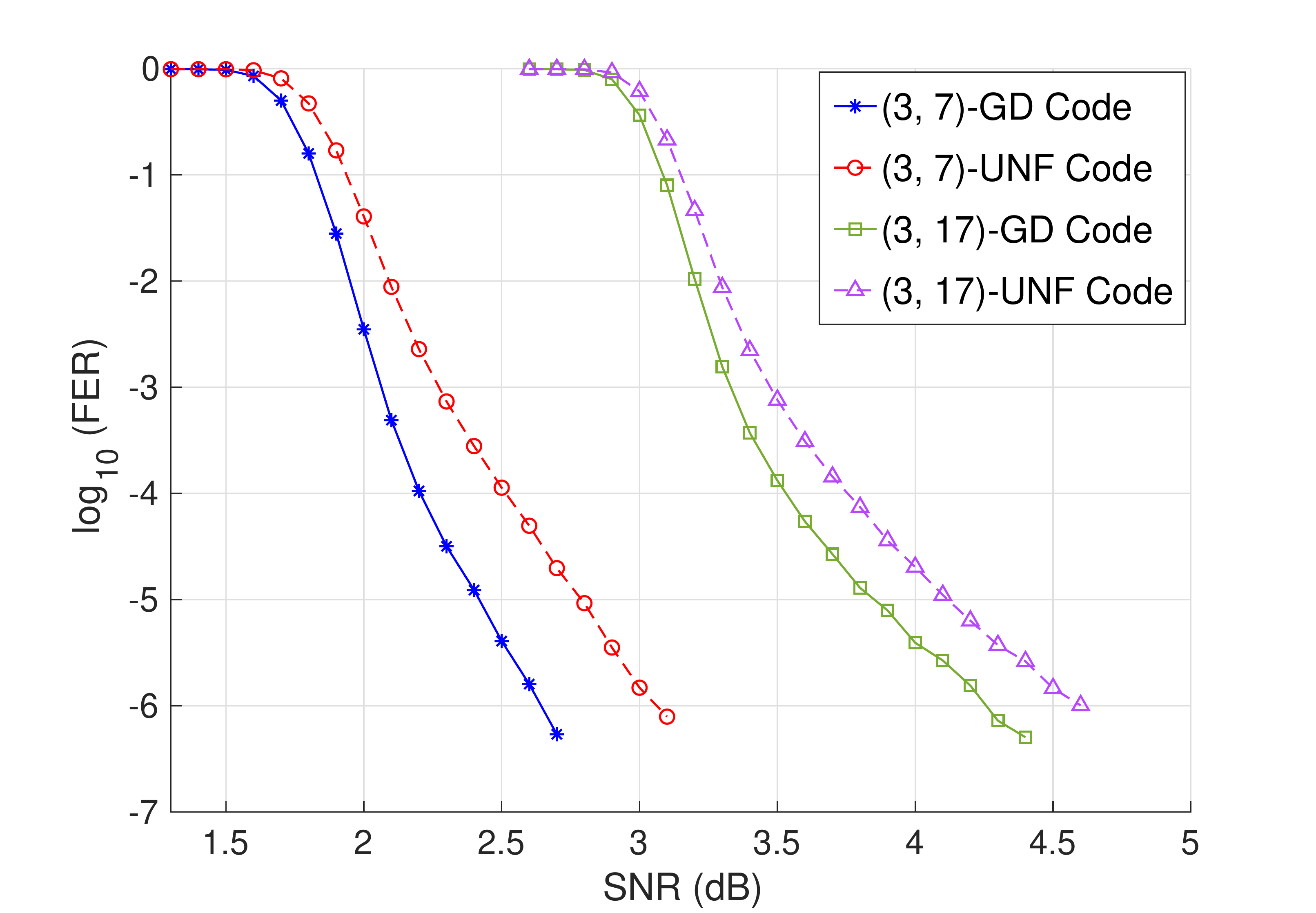}
\caption{FER curves of GD/UNF codes with $\gamma=3$.}
\label{fig: FER_GD_UNF_3}
\end{figure}

\begin{figure}
\centering
\includegraphics[width=0.37\textwidth]{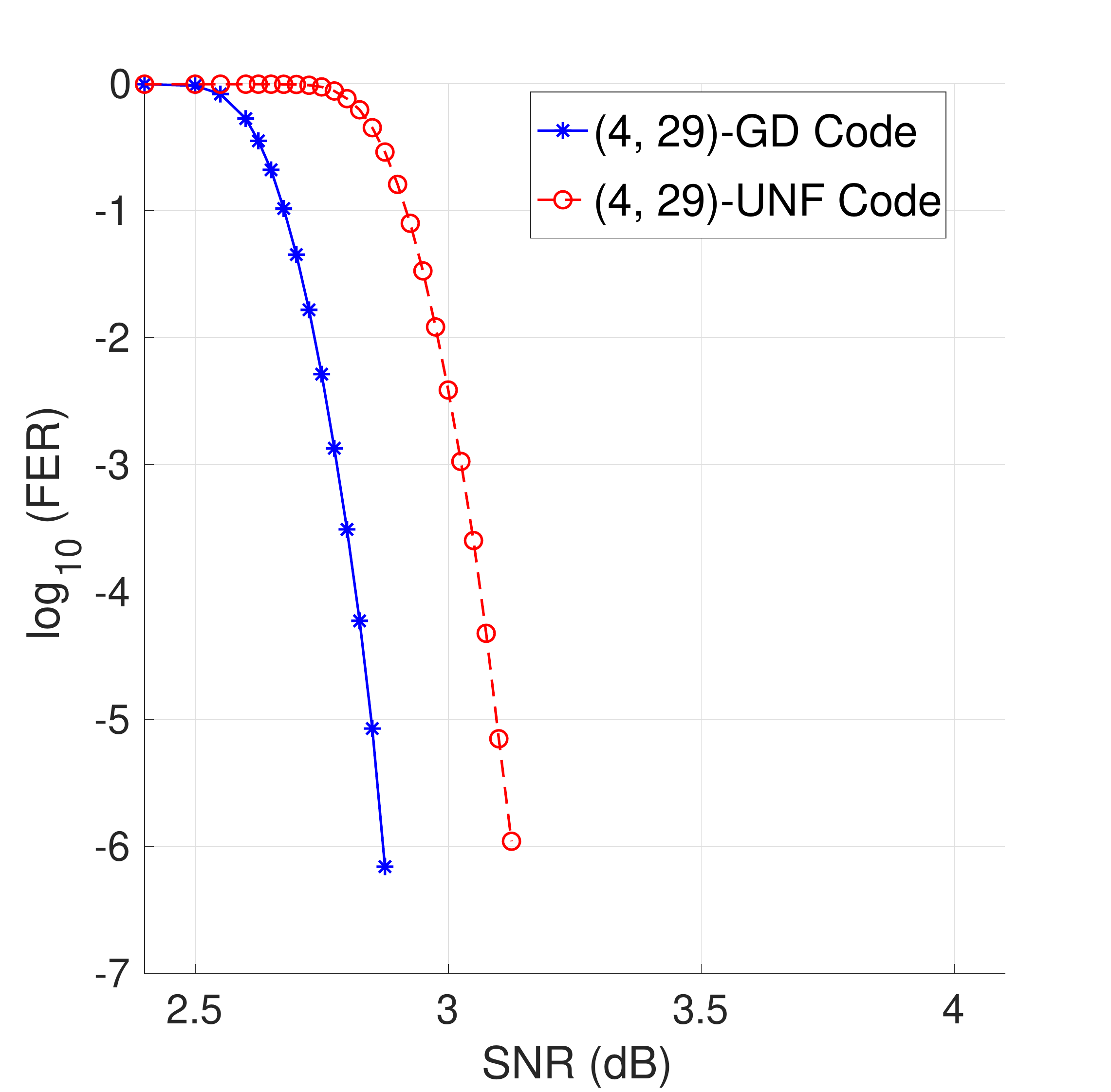}
\caption{FER curves of GD/UNF codes with $(\gamma,\kappa)=(4,29)$.}
\label{fig: FER_GD_UNF_4_29}
\end{figure}

\begin{figure}
\centering
\includegraphics[width=0.42\textwidth]{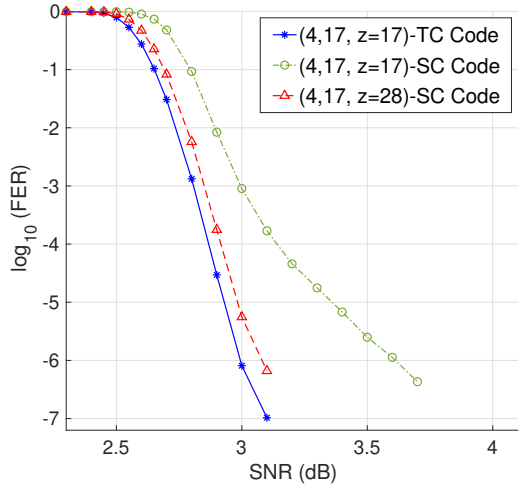}
\caption{FER curves of TC/SC codes with $(\gamma,\kappa)=(4,17)$.}
\label{fig: FER_TC_SC_4_17}
\end{figure}

\begin{figure*}
\centering
\resizebox{0.87\textwidth}{!}{\begin{tabular}{|c|c|c|c|c|c|c|c|c|c|c|c|c|c|c|c|c|c|c|c|c|c|c|c|c|c|c|c|c|}
\hline
0&0&0&19&17&17&1&11&13&5&18&10&19&13&1&6&19&8&19&0&19&0&0&0&2&17&6&19&4\\
\hline
19&18&18&2&0&19&19&5&3&19&9&2&9&9&3&17&6&0&2&16&12&13&8&18&16&0&17&10&0\\
\hline
1&14&3&16&7&1&4&19&5&0&0&16&0&0&7&19&10&19&16&18&3&18&15&3&19&8&19&1&15\\
\hline
16&0&14&1&11&2&15&2&19&16&18&0&19&19&19&0&0&5&1&0&9&4&19&14&7&12&0&19&1\\
\hline
\end{tabular}}
\vspace{5pt}

\resizebox{0.87\textwidth}{!}{\begin{tabular}{|c|c|c|c|c|c|c|c|c|c|c|c|c|c|c|c|c|c|c|c|c|c|c|c|c|c|c|c|c|}
\hline
7&1&18&21&5&4&17&0&6&16&26&8&13&7&5&6&9&2&0&0&0&0&0&5&0&4&19&0&0\\
\hline
3&15&4&1&5&3&12&19&10&21&19&3&4&19&28&1&3&5&12&18&11&10&15&17&19&21&18&25&27\\
\hline
0&8&16&24&3&11&20&20&6&14&22&1&9&2&25&21&12&7&28&6&15&23&19&10&0&1&4&13&21\\
\hline
0&18&7&25&1&3&21&10&28&17&6&24&13&2&20&9&27&16&5&23&12&1&19&8&26&15&4&22&11\\
\hline
\end{tabular}}

\caption{Partitioning matrix (top) and circulant power matrix (bottom) for GD code with $(\gamma,\kappa,m,z,L)=(4,29,19,29,20)$.}
  \label{fig: code_GD_4_29} 
\end{figure*}

\begin{figure*}
\centering
\resizebox{0.87\textwidth}{!}{\begin{tabular}{|c|c|c|c|c|c|c|c|c|c|c|c|c|c|c|c|c|c|c|c|c|c|c|c|c|c|c|c|c|}
\hline
0&17&3&13&1&14&4&12&4&15&10&2&17&2&18&11&17&15&11&3&13&12&13&6&2&5&14&13&14\\
\hline
8&0&19&19&18&8&5&18&13&6&11&3&2&4&11&3&9&15&16&7&7&12&19&16&4&9&0&13&3\\
\hline
14&6&12&10&12&1&17&9&7&5&16&19&1&15&5&19&6&5&15&7&0&2&3&10&15&9&6&7&11\\
\hline
17&14&0&2&9&18&12&1&8&11&4&4&7&10&1&8&14&8&0&16&17&16&1&0&10&18&18&10&8\\
\hline
\end{tabular}}
\vspace{5pt}

\resizebox{0.87\textwidth}{!}{\begin{tabular}{|c|c|c|c|c|c|c|c|c|c|c|c|c|c|c|c|c|c|c|c|c|c|c|c|c|c|c|c|c|}
\hline
12&1&1&7&14&27&4&26&25&2&0&6&15&7&24&1&1&6&17&5&13&19&2&0&11&0&0&0&1\\
\hline
5&2&4&22&8&5&23&1&4&18&28&1&19&17&22&6&3&3&14&9&11&13&15&3&0&2&3&25&27\\
\hline
23&8&2&24&3&7&1&27&6&14&21&12&9&17&5&4&12&20&28&7&7&13&2&25&18&26&5&13&21\\
\hline
28&18&7&4&14&3&21&10&28&17&6&24&13&2&9&8&1&26&5&23&12&1&19&8&26&15&4&22&11\\
\hline
\end{tabular}}

\caption{Partitioning matrix (top) and circulant power matrix (bottom) for UNF code with $(\gamma,\kappa,m,z,L)=(4,29,19,29,20)$.}
  \label{fig: code_UNF_4_29} 
\end{figure*}

Fig.~\ref{fig: FER_GD_UNF_4_29} shows FER curves of the GD/UNF comparison with $(\gamma,\kappa)=(4,29)$. The partitioning matrices and the lifting matrices of the codes are specified in Fig.~\ref{fig: code_GD_4_29}-\ref{fig: code_UNF_4_29}. Cycles-$6$ in the GD code and the UNF code are both removed, and the number of cycles-$8$ in the GD code demonstrates a $51.4\%$ reduction from the count observed in the UNF code. It is worth mentioning that both codes have no ASs of weights up to $8$, which is reflected in their FER curves via the sharp waterfall regions and the non-existing error floor regions. The FER of the GD/UNF codes decreases with a rate exceeding $12$ orders of magnitude per $0.5$ dB. Moreover, the GD code has a significant gain of about $0.25$ dB over the UNF code. These results substantiate the significant potential of the GRADE-AO method in constructing SC codes with superior performance for storage devices, with further applications including wireless communication systems.

\begin{figure*}
\centering
  \subfigure[TC code with $(z,L)=(4,17,2,17,50)$ and $\mathbf{a}=(0,1,4)$.]{%
        \includegraphics[width=0.46\textwidth]{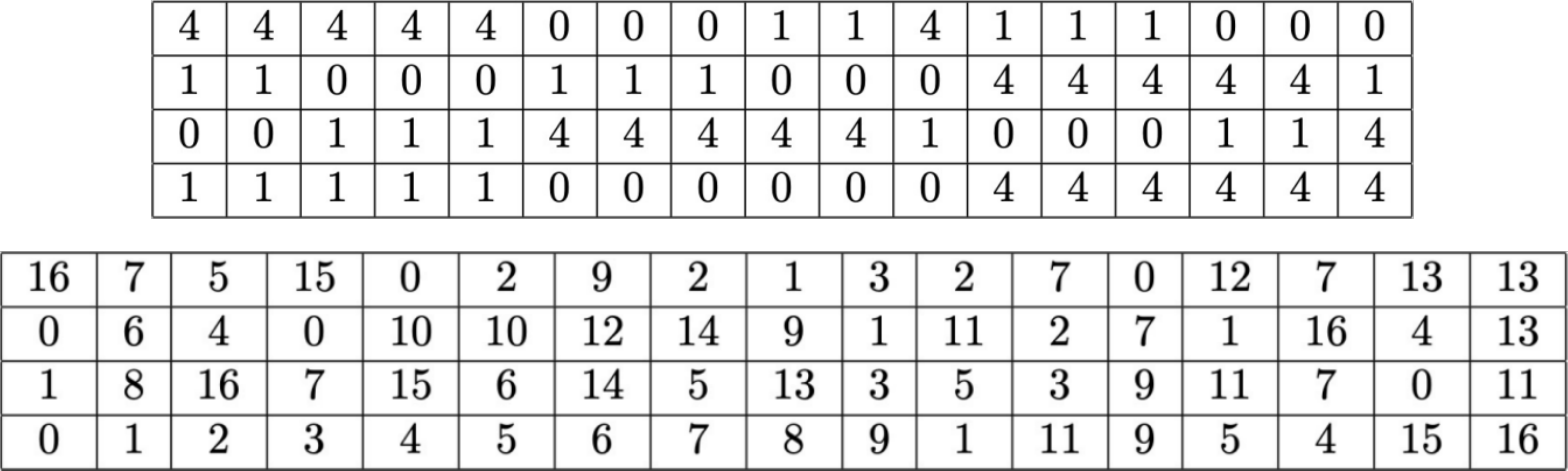}
        \label{1b}
        }
  \subfigure[SC codes with $(z,L)=(17,50)$ (middle) and $(z,L)=(28,30)$ (bottom).]{%
        \includegraphics[width=0.5\textwidth]{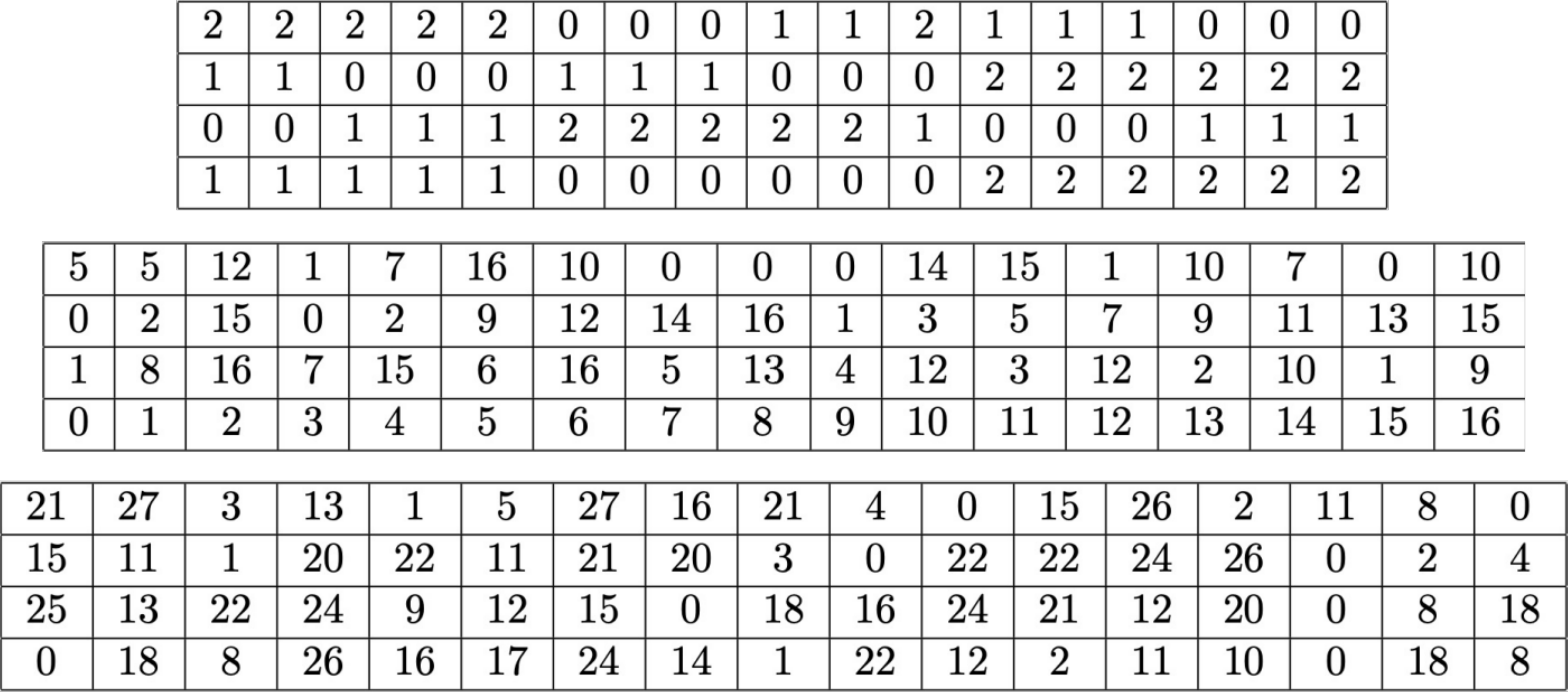}
        \label{1d}
        }
  \caption{Partitioning matrices (top) and circulant power matrices (bottom) for TC/SC codes with $(\gamma,\kappa,m_t)=(4,17,2)$.}
  \label{fig: code_TC_SC_4_17} 
\end{figure*}

Fig.~\ref{fig: FER_TC_SC_4_17} shows the FER curves of the TC/SC codes with $(\gamma,\kappa)=(4,17)$. The partitioning matrices and the lifting matrices of the codes are specified in Fig.~\ref{fig: code_TC_SC_4_17}. The number of cycles-$6$ in the $(4,17)$ TC code demonstrates a $79\%$ and a $20\%$ reduction from the counts observed in the SC codes with a matched constraint length and a matched circulant size, respectively. Moreover, the TC code has no weight-$6$ nor weight-$8$ ASs. It is shown that the TC code outperforms the optimal SC code with a matched constraint length, and that the gain is of greater magnitude when compared with the SC code of identical circulant size. Note that although TC codes have higher memories and thus larger constraint lengths than SC codes of matched circulant sizes, they possess the same number of nonzero component matrices, and thus the same degrees of freedom in construction. This fact makes TC codes even more promising if we can devise for them windowed decoding algorithms with window sizes that are comparable to the corresponding SC codes of matched circulant sizes.

\section{Conclusion}
\label{section: conclusion}
Discrete optimization of the constructions of spatially-coupled (SC) codes with high memories is known to be computationally expensive. Algorithmic optimization is efficient, but can hardly guarantee the performance because of the lack of theoretical guidance. In this paper, we proposed a so-called GRADE-AO method, a probabilistic framework that efficiently searches for locally optimal QC-SC codes with arbitrary memories. We obtain a locally optimal edge distribution that minimizes the expected number of cycle candidates by gradient descent. Starting from a random partitioning matrix with the derived edge distribution, we use algorithmic optimization to find a locally optimal partitioning matrix near it. Simulation results show that our proposed constructions have a significant performance gain over state-of-the-art codes. Future work includes extending the framework on cycle optimization into a one that focuses on detrimental objects.

\section*{Acknowledgment}
This work was supported in part by UCLA Dissertation Year Fellowship, NSF under the Grants CCF-BSF 1718389, CCF 1717602, CCF 2008728, and CCF 1908730, and in part by AFOSR under the Grant 8750-20-2-0504.

\balance
\bibliography{ref}
\bibliographystyle{IEEEtran}

\end{document}